\documentclass{llncs}
\usepackage{amssymb}
\usepackage{amsmath}
\usepackage{color}

\allowdisplaybreaks

\newcommand{\cS}{\mathcal S}
\newcommand{\cW}{\mathcal W}
\newcommand{\cM}{\mathcal M}
\newcommand{\cPS}{\mathcal{PS}}

\newcommand{\F}{{\mathbb F}}
\newcommand{\fp}{{\mathbb F}_{p}}
\newcommand{\fq}{{\mathbb F}_{q}}

\newcommand{\fqk}{{\mathbb F}_{q^k}}
\newcommand{\fpn}{{\mathbb F}_{p^n}}

\newcommand{\fpm}{{\mathbb F}_{p^m}}

\newcommand{\fthreen}{{\mathbb F}_{3^n}}

\newcommand{\Tr}{{\operatorname{Tr}}}
\newcommand{\N}{{\operatorname{N}}}

\begin{document}

\pagestyle{plain}

\title{$p$-ary (Non Weakly) Regular Bent Functions in Generalized $\cPS$ class$^\dagger$}
\titlerunning{Ternary Binomial Bent Functions}

\renewcommand{\thefootnote}{$\dagger$}
\footnotetext{This work was supported by the Research Council of Norway.}

\author{Alexander Kholosha \and Mohit Pal}

\institute{The Selmer Center\\
Department of Informatics, University of Bergen\\
P.O. Box 7800, N-5020 Bergen, Norway\\
\email{oleksandr.kholosha@uib.no, mathmohit@outlook.com}}

\maketitle

\thispagestyle{plain}

\pagestyle{plain}

\begin{abstract}
We use $m$-dimensional partial spreads of $\fp^{2n}$, where $p$ is an odd prime, $n$ is a positive integer and $m$ divides $n$, to 
construct two classes of bent functions from $\fp^{2n}$ to $\fp$. Our construction generalizes the classes of $p$-ary $\cPS^{-}$ and 
$\cPS^{+}$ regular bent functions proposed by P. Lison\v ek and H. Y. Lu (Des. Codes Cryptogr. 73 (2014), 209--216). Unlike the original 
class, such bent functions can be either regular or non weakly regular. 
\end{abstract}

\section{Introduction}
 \label{sec:intro}
Boolean bent functions were first introduced by Rothaus in 1976 as an interesting combinatorial object. They are functions from $\F_2^n$ to 
$\F_2$ for $n$ even, whose Hamming distance to the set of all affine functions is maximum. Equivalently, absolute square of their Walsh 
transform coefficients is constant. They are characterized by the property that their derivatives in nonzero directions are balanced 
(functions having this property are called {\em perfect nonlinear}). Later, the research in this area was stimulated by the significant 
relation to the topics in mathematics and computer science. Kumar, Scholtz and Welch in \cite{KuScWe85} generalized the notion of Boolean 
bent functions to the case of functions from $(\mathbb Z/q\mathbb Z)^n$ to $\mathbb Z/q\mathbb Z$, where $q$ is the power of a prime $p$ 
(which covers the case of functions from $\F_p^n$ to $\F_p$, called {\em $p$-ary functions}, which are the subject of the present paper). 
The complete classification of bent functions looks hopeless even in the binary case, and the case of $p$-ary bent functions is still more 
complicated. However, many explicit methods are known for constructing bent functions, either from scratch (with so-called primary 
constructions) or from other bent functions (in so-called secondary constructions). 

Let $\F_{p^n}$ be the finite field with $p^n$ elements, where $p$ is an odd prime and $n$ is a positive integer. It is well known that 
$\F_{p^n}$ can be endowed with a structure of an $n$-dimensional vector space over $\F_p$. To emphasize the vector space structure of the 
field $\F_{p^n}$, we use the notation $\F_p^n$. Given a function $f$ that maps $\F_p^n$ to $\F_p$, the {\em Walsh transforms} of $f$ as a 
complex-valued function on $\F_p^n$ and its inverse are, respectively, defined by the following identities 
\[\cW_f(y)=\sum_{x\in\F_p^n}\zeta_p^{f(x)-\langle x,y\rangle}\quad\mbox{and}\quad\zeta_p^{f(x)}=\frac{1}{p^n}\sum_{y\in\F_p^n}\cW_f(y)\zeta_p^{\langle x,y\rangle},\]
where $\zeta_p = e^{\frac{2\pi i}{p}}$ is the complex primitive $p^{\rm{th}}$ root of unity, $\langle\cdot,\cdot\rangle$ is the standard 
inner product, and where the elements of $\F_p$ are considered as integers modulo $p$. We call $\cW_f(y)$ the Walsh coefficient of $f$ at 
$y$. 

A {\em $p$-ary function} $f$ is called {\em bent} if all its Walsh transform coefficients satisfy $|W_f(y)|^2=p^n$. According to 
\cite[Property~8]{KuScWe85}, it can be readily seen that for a $p$-ary bent function $f$ with odd $p$ holds 
\begin{equation}
 \label{eq:bent_Wtc}
p^{-n/2}W_f(y)=\left\{\begin{array}{ll}
\pm\,\zeta_p^{f^*(y)},&\ \mbox{if}\ n\ \mbox{is even or}\ n\ \mbox{is odd and}\ p\equiv 1\pmod 4\\
\pm\,i\,\zeta_p^{f^*(y)},&\ \mbox{if}\ n\ \mbox{is odd and}\ p\equiv 3\pmod 4\enspace,
\end{array}\right.
\end{equation}
where $i$ is a complex primitive fourth root of unity. We call $f^*:\fpn\rightarrow\fp$ the {\it dual} of $f$. A $p$-ary bent function $f$ 
is called {\em weakly regular} if we can fix a complex $u$ (called \emph{sign}) having unit magnitude such that 
$p^{-n/2}W_f(y)=u\zeta_p^{f^*(y)}$ for every $y\in\fpn$ (then by (\ref{eq:bent_Wtc}), $u=\pm 1$ or $u=\pm i$). If $u=1$, we call $f$ {\em 
regular} (here necessarily $n$ is even or $p\equiv 1\pmod 4$). 

The dual of a (weakly) regular bent function is again a (weakly) regular bent function. But the dual of non weakly regular bent functions 
is not necessarily bent and we call a bent function {\em dual-bent} if its dual is bent, so all weakly regular bent functions are 
dual-bent. The dual of a non weakly regular dual-bent function is again non weakly regular bent (see \cite{OzPe20}). Both dual-bent and not 
dual-bent non weakly regular bent functions exist, see \cite{CeMeiPo13_1,CeMeiPo16}. 

Two major primary constructions of bent functions are found in the original Maiorana-McFarland and the partial spread classes. The original 
Maiorana-McFarland bent functions from $\F_{p^n} \times \F_{p^n}$ to $\F_p$ are functions of the form 
\begin{equation}
 \label{oMM}
f(x,y) = \Tr_n(y\pi(x)) + h(x)\enspace,
\end{equation}
where $\Tr_n$ is the absolute trace function from $\F_{p^n}$ to $\F_p$, $\pi$ is a permutation of $\F_{p^n}$ and $h$ is a function from 
$\F_{p^n}$ to $\F_p$. It is well known that $\pi$ being a permutation is necessary and sufficient for $f$ to be bent. Any function 
$f:\F_{p^n}\times\F_{p^n}\mapsto\F_p$ that is affine on the cosets of a $n$-dimensional subspace $U$ of $\F_p^{2n}$ is affine equivalent to 
a function of the form \eqref{oMM}. We then say that $f$ belongs to the complete original class of Maiorana-McFarland bent functions with 
associated subspace $U$. 

An \emph{$m$-spread} of $\F_p^n$ is a set of pairwise disjoint (except for $0$) $m$-dimensional subspaces of $\F_p^n$ whose union equals 
$\F_p^n$. It is well known that an $m$-spread of $\F_p^n$ exists if and only if $m$ divides $n$. A \emph{partial $m$-spread} of $\F_p^n$ is 
a set of pairwise disjoint (except for $0$) $m$-dimensional subspaces of $\F_p^n$. In \cite{Di74}, Dillon used partial $n$-spreads of 
$\F_2^{2n}$ to construct two families of bent functions from $\F_2^{2n}$ to $\F_2$, which are now known as $\cPS^{-}$ and $\cPS^{+}$ 
classes of bent functions. In \cite{LiLu14}, Lison\v{e}k and Lu generalized the classes $\cPS^{-}$ and $\cPS^{+}$ of the Boolean bent 
functions to $p$-ary functions and proved combinatorial characterizations of two classes of bent functions from $\F_p^{2n}$ to $\F_p$ using 
partial $n$-spread of $\F_p^{2n}$. The authors called these classes of bent functions as $p$-ary $\cPS^{-}$ and $\cPS^{+}$ class, 
respectively. Citing Carlet \cite[p.~213]{Ca20} ``The situation with ${\cPS}$ is then similar to the situation with general bent functions: 
we have a nice and simple definition, but no systematic way of determining all the elements that satisfy it".

Function $f:\fp^n\mapsto\fp$ is called {\em $m$-normal}, if there exists an $m$-dimensional affine subspace of $\fp^n$ restricted to which 
$f$ is constant. If $f$ is affine on an $m$-dimensional affine subspace of $\fp^n$, then $f$ is called weakly $m$-normal. When $n$ is even 
and $m=n/2$, then $f$ is called (weakly) normal. Classes of completed Maiorana-McFarland functions and of the $\cPS^{+}$ bent functions are 
normal by definition. 

In this paper, we generalize the results of Lison\v{e}k and Lu by taking partial $m$-spreads of $\F_p^{2n}$, where $m$ divides $n$, and 
propose two classes of bent functions mapping $\F_p^{2n}$ to $\F_p$. These classes contain both regular and non weakly regular bent 
functions. The paper is organized as follows. In Section~\ref{sec:pre}, we recall some basic notions and related results that will be used 
in the subsequent sections. In Sections~\ref{sec:PSminus} and \ref{sec:PSplus}, we propose two classes of bent functions from $\F_p^{2n}$ 
to $\F_p$ using partial $m$-spreads of $\F_p^{2n}$ that generalize $\cPS$ construction. Unlike the original $\cPS^{-}$ and $\cPS^{+}$ 
classes, such bent functions can also be non weakly regular. In the final Section~\ref{sec:Di_monom}, we study the question whether a 
monomial function with a generalized Dillon exponent can be bent (but not in the $\cPS^{-}$ class). 

\section{Preliminaries}
 \label{sec:pre}
Let $\F_p^*:=\F_p\setminus\{0\}$ and similarly, for any subspace $T$ of $\F_p^n$ denote $T^*:=T\setminus\{0\}$. For $a\in\F_p^n$, denote 
\[a^{\perp}:=\{x\in\F_p^n:\langle a,x\rangle=0\}\enspace.\]
Throughout the paper, we use following lemma concerning the character sum over a subspace of $\F_p^n$. 

\begin{lemma}
 \label{le:sspsum}
Let $a\in\F_p^n$, $a\neq 0$, and assume that $T$ is a subspace of $\F_p^n$ such that $T^*\not\subset a^{\perp}$. Then 
\[\sum_{x\in T}\zeta_p^{-\langle a,x\rangle}=0\enspace.\]
\end{lemma}

We shall also use the following lemma concerning the value distribution of bent functions from $\F_p^{2n}$ to $\F_p$.

\begin{lemma}[\cite{Ny91}]
 \label{le:Ny}
Let $n$ be a positive integer and $p$ be a prime. Suppose that $f:\F_p^{2n}\rightarrow \F_p$ is a bent function and for $j\in\F_p$ denote 
$D_j:= f^{-1}(j)$. Then there exists $k\in\F_p$ such that 
\begin{align*} 
\lvert D_k\rvert&=p^{2n-1}\pm(p-1)p^{n-1},\\
\lvert D_\ell\rvert&=p^{2n-1}\mp p^{n-1}\ \mbox{for}\ \ell\in\F_p\setminus\{k\}\enspace.
\end{align*}
Here the $\pm$ signs are taken correspondingly. Moreover, for regular bent functions, the upper signs are taken.
\end{lemma}

In \cite{Ho04_1}, Hou showed that a $p$-ary weakly regular bent function $f:\F_p^n \rightarrow \F_p$ with $p^n>3$, can have algebraic 
degree at most $n(p-1)/2$ and the algebraic degree of a non weakly regular bent function is upper bounded by $n(p-1)/2+1$. The following 
result is a generalization of \cite[Theorem~8]{AnMei22} where it was proved that all $p$-ary $\cPS$ bent functions attain the maximal 
algebraic degree (for regular bent functions). Note that we do not need $f$ to be bent here, as was assumed in \cite{AnMei22}. 

\begin{proposition}
 \label{pr:deg}
Let $p$ be an odd prime and $n$ be a positive integer. Assume $f:\F_p^{2n}\rightarrow \F_p$ is such that $f(0,\ldots,0)=0$ and there exists 
an $m$-dimensional subspace $U$ of $\F_p^{2n}$ such that $f$ is a nonzero constant on $U^*$. Then $f$ has algebraic degree 
$\deg(f)\geq(p-1)m$. 
\end{proposition}

\begin{proof}
Let $f(x_1,\ldots,x_{2n})\equiv c$ on $U^*$ for some $c\in\F_p^*$. By a coordinate transformation $A$ of $\F_p^{2n}$, we have 
\[A(U)=\{(\alpha_1,\ldots,\alpha_m,0,\ldots,0)\}=\F_p^m\times\{(0,\ldots,0)\}=:H\enspace.\]
Denoting $g=f\circ A^{-1}$ we obtain $g(x_1,\ldots,x_{2n})\equiv c$ on $H^*$. Also denote 
\[g'(x_1,\ldots,x_{m}):= g(x_1,\ldots,x_{m},0,\ldots,0).\]
Note that $\deg(f)=\deg(g)\geq\deg(g')$. Hence, it suffices to show that $\deg(g')=(p-1)m$. By Lagrange interpolation, we write 
\[g'(x_1,\ldots,x_{m})=\sum_{(\alpha_1,\ldots,\alpha_{m})\in\F_p^m}g'(\alpha_1,\ldots,\alpha_{m})\prod_{i=1}^m (1-(x_i-\alpha_i)^{p-1})\enspace.\]
Since $g'(0,\ldots,0)=0$ and $g'(x_1,\ldots,x_{m})\equiv c$ on $H^*$, we have 
\begin{align*}
g'(x_1,\ldots,x_{m}) &= c\left( 1- \prod_{i=1}^m (1-x_i^{p-1}) \right)\\
&= (-1)^{m+1} c\cdot x_1^{p-1}\cdots x_m^{p-1}+h(x_1,\ldots,x_m),
\end{align*}
for some $h\in\F_p[x_1,\ldots,x_{m}]$ of degree less than $(p-1)m$. Hence, $g'$ has algebraic degree $(p-1)m$ and $f$ has algebraic degree 
at least $(p-1)m$.\qed 
\end{proof}

\begin{corollary}
Let $p$ be an odd prime and $n$ be a positive integer. Let $f:\F_p^{2n}\rightarrow \F_p$ be a function having algebraic degree less than 
$(p-1)m$ for some $m\leq 2n$. Then, $f$ cannot be constant on $U^*$ for any $m$-dimensional subspace $U$ of $\F_p^{2n}$ such that $f(U^*) 
\equiv c \in \F_p$ and $c \neq f(0)$. 
\end{corollary}

\section{$p$-ary $m$-$\cPS^{-}$ class}
 \label{sec:PSminus}
In \cite{LiLu14}, Lison\v{e}k and Lu proposed two classes of bent functions from $\F_p^{2n}$ to $\F_p$, namely $p$-ary $\cPS^{-}$ and 
$\cPS^{+}$ classes. Moreover, the authors showed that these bent functions are regular. It is easy to see from Proposition~\ref{pr:deg} 
that $p$-ary $\cPS^{-}$ and $\cPS^{+}$ classes of bent functions from $\F_p^{2n}$ to $\F_p$ have the maximum possible algebraic degree 
$(p-1)n$. Experimental results suggest that there are $p$-ary bent functions from $\F_p^{2n}$ to $\F_p$ of algebraic degree strictly less 
than $(p-1)n$ that are constant on the punctured subspaces $U$ of dimension $m<n$ of $\F_p^{2n}$. Following is a generalization of their 
Theorem~3.3 that introduced $p$-ary $\cPS^{-}$ class, we suggest a method for constructing bent functions $f:\F_p^{2n}\rightarrow\F_p$ 
using partial $m$-spreads of $\F_p^{2n}$. 

\begin{theorem}
 \label{th:PSminus}
Let $p$ be an odd prime and $m,n$ be positive integers such that $m<2n$ and $p^m>3$. Assume $\cS$ is a partial $m$-spread of $\F_p^{2n}$ 
and $f:\F_p^{2n}\rightarrow\F_p$ is such that for each $T\in\cS$, $f$ is constant on $T^*$. Moreover, denote $D_j=f^{-1}(j)$ for $j\in\F_p$ 
and suppose that $D_0=\F_p^{2n}\setminus\bigcup_{T\in\cS}T^*$. Then $f$ is bent if and only if following three conditions hold. 
\begin{enumerate}
\item $m\mid n$
\item $D_j\cup\{0\}$ is a union of exactly $\frac{p^{n-1}(p^n-1)}{p^m-1}$ elements of $\cS$ for every $j\in\F_p^*$
\item for every nonzero $a\in\F_p^{2n}$, $a^{\perp}$ contains
\begin{align}
 \label{eq:bentcond1}
&N\ \mbox{subspaces}\ T\subseteq D_j\cup\{0\} \ \mbox{for every}\ j\in\F_p^*\ \mbox{or}\\
 \label{eq:nwrbentcond1}
&M\ \mbox{subspaces}\ T\subseteq D_j\cup\{0\} \ \mbox{for every}\ j\in\F_p^*\ \mbox{or}\\
 \label{eq:bentcond2}
&\begin{cases}
N+p^{n-m}\ \mbox{subspaces}\ T\subseteq D_t\cup\{0\}\ \mbox{for some}\ t\in \F_p^*\ \mbox{and}\\
N\ \mbox{subspaces}\ T\subseteq D_j\cup\{0\} \ \mbox{for remaining}\ j\in\F_p^*\setminus\{t\}\ \mbox{or}
\end{cases}\\
 \label{eq:nwrbentcond2}
&\begin{cases}
M-p^{n-m}\ \mbox{subspaces}\ T\subseteq D_t\cup\{0\}\ \mbox{for some}\ t\in \F_p^*\ \mbox{and}\\
M\ \mbox{subspaces}\ T\subseteq D_j\cup\{0\}\ \mbox{for remaining}\ j\in\F_p^*\setminus\{t\}\enspace,
\end{cases}
\end{align}
where $N=\frac{p^{n-1}(p^{n-m}-1)}{p^m-1}$ and $M=\frac{p^{n-m-1}(p^{n}+p^m-2)}{p^m-1}$ (here assume $m\neq n$). Also, $\cW_f(0)=p^n$, if 
\eqref{eq:bentcond1} holds then $\cW_f(a)=p^n$, if \eqref{eq:nwrbentcond1} holds then $\cW_f(a)=-p^n$, if \eqref{eq:bentcond2} holds then 
$\cW_f(a)=p^n\zeta_p^t$, and if \eqref{eq:nwrbentcond2} holds then $\cW_f(a)=-p^n\zeta_p^t$. Therefore, $f$ is bent then it is either 
regular or non weakly regular bent function and the latter holds only if $m\neq n$ and \eqref{eq:nwrbentcond1} or \eqref{eq:nwrbentcond2} 
is satisfied for at least one nonzero $a\in\F_p^{2n}$. 
\end{enumerate}
\end{theorem}

\begin{proof}
For $j\neq 0$ let $N_j\geq 0$ be such that $D_j=\bigcup_{i=1}^{N_j}\cS_{ji}^*$, where $\cS_{ji}$ are pairwise distinct (except $0$) 
elements of $\cS$ (if $N_j=0$ then $D_j=\emptyset$). Therefore, 
\begin{align*}
\lvert D_j\rvert&=N_j(p^m-1)\ \mbox{for}\ j\neq 0\\
\lvert D_0\rvert&=p^{2n}-(p^m-1)\sum_{j=1}^{p-1}N_j\enspace. 
\end{align*}
Note that $p-1$ divides $\lvert D_j\rvert$ when $j\neq 0$ and does not divide $\lvert D_0\rvert$.

First, assume that $f$ is bent. Considering Lemma~\ref{le:Ny}, and the mentioned above divisibility, we conclude that $k=0$. By the same 
lemma, $\lvert D_j\rvert=p^{n-1}(p^n\mp 1)$ for $j\neq 0$. For such $j$, $\lvert D_j\rvert$ is divisible by $p^m-1$ and since 
$\gcd(p^n+1,p^m-1)\neq p^m-1$ (unless $p^m=3$), it follows that $\lvert D_j\rvert=p^{n-1}(p^n-1)$, $m\mid n$ and 
$N_j=\frac{p^{n-1}(p^n-1)}{p^m-1}$. 

Now, it is easy to see that
\begin{align*}
\cW_f(0)&=\sum_{x\in\F_p^{2n}}\zeta_p^{f(x)}=\sum_{j=0}^{p-1}\sum_{x\in D_j}\zeta_p^j=\lvert D_0 \rvert+\sum_{j=1}^{p-1}\lvert D_j\rvert\zeta_p^j\\
&=p^{2n}-(p^m-1)\sum_{j=1}^{p-1}N_j+(p^m-1)\sum_{j=1}^{p-1}N_j\zeta_p^j\\
&=p^{2n}-(p-1)p^{n-1}(p^n-1)-p^{n-1}(p^n-1)=p^n\enspace.
\end{align*}
From here, we conclude that $f$ is either regular or non weakly regular bent function. Now we show that \eqref{eq:bentcond1}, 
\eqref{eq:nwrbentcond1}, \eqref{eq:bentcond2} or \eqref{eq:nwrbentcond2} holds. Let $a\in\F_p^{2n}$ with $a\neq 0$. Assume that $a^{\perp}$ 
contains $s$ distinct subspaces from $\cS$, say $\cS_{j_1 i_1},\ldots ,\cS_{j_s i_s}$, where $0<j_1,\ldots,j_s<p$ and $1\leq 
i_1,\ldots,i_s\leq\frac{p^{n-1}(p^n-1)}{p^m-1}$. Consider 
\begin{align} \label{eq:T1wa}
\nonumber \cW_f(a)&=\sum_{x\in\F_p^{2n}}\zeta_p^{f(x)-\langle a,x\rangle}\\
\nonumber &=\sum_{x\in D_0}\zeta_p^{-\langle a,x\rangle}+\sum_{j=1}^{p-1}\sum_{x \in D_j}\zeta_p^{f(x)-\langle a,x \rangle}\\
\nonumber &=-\sum_{j=1}^{p-1}\sum_{x\in D_j}\zeta_p^{-\langle a,x \rangle}+\sum_{j=1}^{p-1}\zeta_p^j\sum_{x\in D_j}\zeta_p^{-\langle a, x\rangle}\\
\nonumber &=-\sum_{j=1}^{p-1}\sum_{i=1}^{\frac{p^{n-1}(p^n-1)}{p^m-1}}\sum_{x\in\cS_{ji}^*}\zeta_p^{-\langle a,x\rangle}+\sum_{j=1}^{p-1}\zeta_p^j \sum_{i=1}^{\frac{p^{n-1}(p^n-1)}{p^m-1}}\sum_{x\in\cS_{ji}^*}\zeta_p^{-\langle a,x\rangle}\\
\nonumber &=(p-1)\frac{p^{n-1}(p^n-1)}{p^m-1}-sp^m +\frac{p^{n-1}(p^n-1)}{p^m-1}+p^m(\zeta_p^{j_1}+\cdots+\zeta_p^{j_s})\\
&=\frac{p^{n}(p^n-1)}{p^m-1}+p^m(\zeta_p^{j_1}+\cdots+\zeta_p^{j_s}-s)\enspace.
\end{align}
Here, we used Lemma~\ref{le:sspsum} consistently. Since we are in the even dimension $2n$ and $f$ is either regular or non weakly regular 
bent function, from \eqref{eq:bent_Wtc}, we have $\cW_f(a)= \pm p^n \zeta_p^t$ for some $t \in \F_p$. We shall consider following four. 

\textbf{Case 1.} Let $\cW_f(a)= p^n$. In this case, \eqref{eq:T1wa} reduces to
\begin{align}\label{eq:T1C1}
\zeta_p^{j_1}+\cdots +\zeta_p^{j_s}=s+p^{n-m}-\frac{p^{n-m}(p^n-1)}{p^m-1}\enspace.
\end{align}
The LHS of \eqref{eq:T1C1} is an integer if and only if it is of the form $\ell(\zeta_p+\cdots+\zeta_p^{p-1})=-\ell$ for some positive 
integer $\ell$. This implies $s=\ell(p-1)$. Inserting this value in \eqref{eq:T1C1} we obtain 
\[-\ell=\ell(p-1)+p^{n-m}-\frac{p^{n-m}(p^n-1)}{p^m-1}\implies \ell=\frac{p^{n-1}(p^{n-m}-1)}{p^m-1}=N\enspace.\]
We conclude that $a^{\perp}$ contains $N$ distinct $\cS_{ji}\subseteq D_j \cup \{0\} $ for every $j\in\F_p^*$. Which is precisely condition 
\eqref{eq:bentcond1}. 

\textbf{Case 2.} Let $\cW_f(a)= -p^n$. In this case, \eqref{eq:T1wa} reduces to
\begin{align}\label{eq:T1C2}
\zeta_p^{j_1}+\cdots +\zeta_p^{j_s}=s-p^{n-m}-\frac{p^{n-m}(p^n-1)}{p^m-1}\enspace.
\end{align}
The LHS of \eqref{eq:T1C2} is an integer if and only if it is of the form $\ell(\zeta_p+\cdots+\zeta_p^{p-1})=-\ell$ for some positive 
integer $\ell$. This implies $s=\ell(p-1)$. Inserting this value in \eqref{eq:T1C2} we obtain 
\[
-\ell=\ell(p-1)-p^{n-m}-\frac{p^{n-m}(p^n-1)}{p^m-1}\implies \ell=\frac{p^{n-1-m}(p^n+p^m-2)}{p^m-1}=M\enspace.
\]
It is easy to see that $M$ is a integer if and only if $m<n$. We conclude that $a^{\perp}$ contains $M$ distinct $\cS_{ji}\subseteq D_j \cup \{0\} $ 
for every $j\in\F_p^*$. Which is precisely condition \eqref{eq:nwrbentcond1}. 

\textbf{Case 3.} Let $\cW_f(a)= p^n\zeta_p^t$ for some $t \in \F_p^*$. In this case, \eqref{eq:T1wa} reduces to
\begin{align}\label{eq:T1C3}
\zeta_p^{j_1}+\cdots+\zeta_p^{j_s}-p^{n-m}\zeta_p^t=s-\frac{p^{n-m}(p^n-1)}{p^m-1}\enspace.
\end{align}
Again, the LHS of \eqref{eq:T1C3} is an integer if and only if it is of the form $\ell(\zeta_p+\cdots+\zeta_p^{p-1})=-\ell$ for some 
positive integer $\ell$. This implies that 
\[
\zeta_p^{j_1}+\cdots +\zeta_p^{j_s}=\ell(\zeta_p+\cdots +\zeta_p^{p-1})+p^{n-m}\zeta_p^t
\]
hence $s=\ell(p-1)+p^{n-m}$. Inserting this value in \eqref{eq:T1C3} we have
\[
-\ell=\ell(p-1)+p^{n-m}-\frac{p^{n-m}(p^n-1)}{p^m-1}\implies \ell=\frac{p^{n-1}(p^{n-m}-1)}{p^m-1}=N\enspace.
\]
We conclude that $a^{\perp}$ contains $N$ distinct $\cS_{ji}\subseteq D_j \cup \{0\} $ for each $j\in\F_p^*\setminus\{t\}$ and $a^{\perp}$ contains 
$N+p^{n-m}$ distinct $\cS_{ti}\subseteq D_t \cup \{0\} $. Which is precisely condition \eqref{eq:bentcond2}. 

\textbf{Case 4.} Let $\cW_f(a)= -p^n\zeta_p^t$ for some $t \in \F_p^*$. In this case, \eqref{eq:T1wa} reduces to
\begin{align}\label{eq:T1C4}
 \zeta_p^{j_1}+\cdots+\zeta_p^{j_s}+p^{n-m}\zeta_p^t=s-\frac{p^{n-m}(p^n-1)}{p^m-1}\enspace. 
\end{align}
In this case, the LHS of \eqref{eq:T1C4} is an integer if and only if it is of the form $\ell(\zeta_p+\cdots+\zeta_p^{p-1})=-\ell$ for some 
positive integer $\ell$. This implies that 
\[
\zeta_p^{j_1}+\cdots +\zeta_p^{j_s}=\ell(\zeta_p+\cdots +\zeta_p^{p-1})-p^{n-m}\zeta_p^t
\]
hence $s=\ell(p-1)-p^{n-m}$. Inserting this value in \eqref{eq:T1C4} we have
\[
-\ell=\ell(p-1)-p^{n-m}-\frac{p^{n-m}(p^n-1)}{p^m-1}\implies \ell=\frac{p^{n-1-m}(p^{n}+p^m-2)}{p^m-1}=M\enspace.
\]
It is easy to see that $M$ is a integer if and only if $m<n$. We conclude that $a^{\perp}$ contains $M$ distinct $\cS_{ji}\subseteq D_j \cup \{0\} $ 
for each $j\in\F_p^*\setminus\{t\}$ and $a^{\perp}$ contains $M-p^{n-m}$ distinct $\cS_{ti}\subseteq D_t \cup \{0\} $. Which is precisely condition 
\eqref{eq:nwrbentcond2}. 

Now assume that $m\mid n$ and $f$ is such that for every $j\in\F_p^*$, 
\[
D_j=\bigcup_{i=1}^{\frac{p^{n-1}(p^n-1)}{p^m-1}}\cS^*_{ji}\enspace,
\]
where $\cS_{ji}$ are pairwise distinct (except $0$) elements of $\cS$. It is easy to see that, under these assumptions, we have 
$\cW_f(0)=p^n$. For any nonzero $a\in\F_p^{2n}$, it is easy to verify that if \eqref{eq:bentcond1} holds then $\cW_f(a)=p^n$, if 
\eqref{eq:nwrbentcond1} holds then $\cW_f(a)=-p^n$, if \eqref{eq:bentcond2} holds then $\cW_f(a)=\zeta_p^t p^n$, and if 
\eqref{eq:nwrbentcond2} holds then $\cW_f(a)=-\zeta_p^t p^n$. 

If all nonzero $a\in\F_p^{2n}$ satisfy either \eqref{eq:bentcond1} or \eqref{eq:bentcond2} then $f$ is regular bent. If all nonzero 
$a\in\F_p^{2n}$ satisfy any of these four conditions and at least one nonzero $a\in\F_p^{2n}$ satisfies \eqref{eq:nwrbentcond1} or 
\eqref{eq:nwrbentcond2} then $f$ non weakly regular bent.\qed 
\end{proof}

The condition for $p^m>3$ is relevant since with $p=3$ and $m=1$, Theorem~\ref{th:PSminus} would claim that all bent functions satisfying 
$f(x)=f(-x)$ (thus, from $1$-$\cPS^{-}$ class) are regular or non weakly regular that is wrong (take, e.g., weakly regular but not regular 
bent binomials $\Tr^n_{2k}\big(x^{p^{3k}+p^{2k}-p^k+1}+x^2\big)$ with any odd $p$ and $n=4k$ from \cite{HeKh10_4}). The following corollary 
is \cite[Theorem~3.3]{LiLu14} and for odd $p$, is a special case of Theorem~\ref{th:PSminus} with $m=n$. The case when $p=2$ is proven in 
\cite[Ch.~6]{Di74}. 

\begin{corollary}
Let $p$ be a prime and $n$ be a positive integer such that $p^n>3$. Assume $\cS$ is a partial $n$-spread of $\fp^{2n}$ and 
$f:\fp^{2n}\mapsto\fp$ is such that for each $T\in\cS$, $f$ is constant on $T^*$. Moreover, denote $D_j=f^{-1}(j)$ for $j\in\F_p$ and 
suppose that $D_0=\fp^{2n}\setminus\bigcup_{T\in\cS}T^*$. Then $f$ is bent if and only if $D_j\cup\{0\}$ is a union of exactly $p^{n-1}$ 
elements of $\cS$ for every $j\in\fp^*$. Moreover, if $f$ is bent then for every nonzero $a\in\fp^{2n}$, $a^{\perp}$ contains either 
\begin{align}
 \label{eq:bentcond3}
&\mbox{none of the subspaces from}\ \cS\ \mbox{or}\\
 \label{eq:bentcond4}
&\mbox{only one subspace}\ T\in\cS\ \mbox{with}\ T\subseteq D_t\cup\{0\}\ \mbox{for some}\ t\in\fp^*
\end{align}
and these are the only cases possible. Also, $\cW_f(0)=p^n$ and if \eqref{eq:bentcond3} holds then $\cW_f(a)=p^n$ and if 
\eqref{eq:bentcond4} holds then $\cW_f(a)=p^n\zeta_p^t$. If $f$ is bent then it is a regular bent functions. 
\end{corollary}

\begin{definition}
The class of bent functions satisfying conditions of Theorem~\ref{th:PSminus} is called $p$-ary $m$-$\cPS^{-}$ class.
\end{definition}

As noted in \cite[Corollary~3.5]{LiLu14}, the dual of a $\cPS^{-}$ function is also in $\cPS^{-}$. It follows since the dual subspaces of 
$T\in\cS$ have the same dimension $n$ and also form a partial $n$-spread. The general case of $m$-$\cPS^{-}$ functions is different. Assume 
$f$ is an $m$-$\cPS^{-}$ regular bent function. For any $\cS_{jl}\in\cS$, denote its dual ($(2n-m)$-dimensional subspace) as 
$\cS_{jl}^{\perp}$. By Theorem~\ref{th:PSminus}, 
\[(f^*)^{-1}(0)=\bigcup_{j=1}^{p-1}\bigcup_{i=1}^N\cS_{jl_i}^{\perp}\enspace,\]
where $1\leq \ell_1<\dots \ell_N\leq T$ and $T=\frac{p^{n-1}(p^n-1)}{p^m-1}$. Also, for $t\neq 0$, 
\[(f^*)^{-1}(t)=\Bigg(\bigcup_{i=1}^{N+p^{n-m}}\cS_{th_i}^{\perp}\Bigg)\bigcup\Bigg(\bigcup_{\substack{j=1\\j\neq t}}^{p-1}\bigcup_{i=1}^N\cS_{jl_i}^{\perp}\Bigg)\enspace,\]
where $1\leq h_1<\dots h_{N+p^{n-m}}\leq T$ and $1\leq\ell_1<\dots \ell_N\leq T$. In this case, we can have $2n-m>n$ and answer to the 
question which class the dual belongs to, is open. 

\begin{example}[\cite{HeKhSp26}]
 \label{ex:PSmin_MMF}
Let $n=2k$. Take any $a_1$ being a nonsquare in $\F_{3^{2n}}$ and define 
\[
a_2=\pm I^k a_1^{(3^k+1)/2}\left((-1)^k a_1^{(3^k-1)(3^{2k}+1)/4}+a_1^{-(3^k-1)(3^{2k}+1)/4}\right)\enspace,
\]
where $I$ is a primitive $4$th root of unity in $\F_{3^{2n}}$ and the sign is arbitrary. Ternary vectorial function 
$F:\F_{3^{2n}}\mapsto\F_{3^k}$ given by 
\[
F(x)=\Tr^{2n}_k\left(a_1 x^{2(3^k+1)}+a_2 x^{(3^k+1)^2}\right)
\]
is a bent function of algebraic degree four. Moreover, component functions of $F$ are all regular bent functions from the completed class 
$\cM$. 

Now assume $k$ is odd. Then $F$ is constant on the elements of the $2$-spread of $\F_{3^{2n}}$ that is a set of pairwise disjoint (except 
for $0$) $2$-dimensional subspaces of $\F_{3^{2n}}$ defined by $\xi^i\F_{3^2}$ for $i=0,\dots,\frac{3^{2n}-1}{8}-1$. Therefore, component 
functions of $F$ are all bent functions also from the $2$-$\cPS^{-}$ class. Moreover, for $n>2$, the algebraic degree of $f$ (that is equal 
four) is not $2n$ that is the maximum for a regular ternary bent function and, by Proposition~\ref{pr:deg}, $f$ is not in the $\cPS^{-}$ 
class (i.e., it is not a $n$-$\cPS^{-}$ function). 
\end{example}

\begin{example}[\cite{HeKhSp26}]
Take $n>1$ not divisible by six. Ternary vectorial function $F:\F_{3^{2n}}\mapsto\fthreen$ given by 
\[F(x)=\Tr^{2n}_n\big(a(x^{2(3^{k+1}+1)}+x^{3^{k+1}+5}+x^8)\big)\enspace,\]
where $a=\xi^{(3^n+1)/2}$ and $\xi$ is a primitive element of $\F_{3^{2n}}$, is a bent function of algebraic degree four. Moreover, 
component functions of $F$ are all regular bent functions from the completed class $\cM$. 

Now assume $n$ is even. Using the same argument as in Example~\ref{ex:PSmin_MMF} we conclude that component functions of $F$ are all bent 
functions also from the $2$-$\cPS^{-}$ but not from the $\cPS^{-}$ class. 
\end{example}

\begin{example} 
Take $\F_{3^8}$ and $\xi$ let be its primitive element. Then $f(x)=\Tr_8(\eta x^{88}+\eta^3 x^{136})$, where $\eta=\xi^{(3^8-1)/8}$, is a 
$2$-$\cPS^{-}$ ternary regular bent function of algebraic degree four, so it is not from the $\cPS^{-}$ class. 

Take $\F_{5^4}$ and $\xi$ let be its primitive element. Then $f(x)=\Tr_4(\eta^{14}x^{28}\pm x^{36})$, where $\eta=\xi^{(5^2+1)/2}$, is a 
$1$-$\cPS^{-}$ quinary regular bent function of algebraic degree four, so it is not from the $\cPS^{-}$ class. 
\end{example}

\begin{example}
 \label{ex:PSmin_nwr}
Define a quinary function $f:\F_5^4\mapsto\F_5$ as follows
\[f(x_1, x_2, x_3, x_4) = x_1x_2^3 + x_2^2(x_3^2+2x_4^2)+(1-x_2^4)x_3x_4^3\enspace.\]
Then, $f$ is a non weakly regular dual-bent function of algebraic degree eight and both $f$ and its dual belong to $1$-$\cPS^{-}$ class. 
\end{example}

\begin{proof}
The Walsh transform of $f$ at $(a_1, a_2, a_3, a_4) \in \F_5^4$ is given by
\begin{align*}
\cW_f(a_1, a_2, a_3, a_4) &= \sum_{x_1} \sum_{x_2} \sum_{x_3} \sum_{x_4} \zeta_5^{x_1x_2^3 + x_2^2(x_3^2+2x_4^2)+(1-x_2^4)x_3x_4^3 - a_1x_1-a_2x_2-a_3x_3-a_4x_4}\\
&= \sum_{x_2} \sum_{x_3} \sum_{x_4} \zeta_5^{x_2^2(x_3^2+2x_4^2)+(1-x_2^4)x_3x_4^3 -a_2x_2-a_3x_3-a_4x_4} \sum_{x_1} \zeta_5^{x_1(x_2^3 - a_1)}\\
&= 5\sum_{x_2^3=a_1} \sum_{x_3} \sum_{x_4} \zeta_5^{x_2^2(x_3^2+2x_4^2)+(1-x_2^4)x_3x_4^3 -a_2x_2-a_3x_3-a_4x_4}\\
&= 5\sum_{x_3} \sum_{x_4} \zeta_5^{a_1^2(x_3^2+2x_4^2)+(1-a_1^4)x_3x_4^3 -a_2a_1^3-a_3x_3-a_4x_4} \\
&= 5\zeta_5^{ -a_2a_1^3} \sum_{x_3} \sum_{x_4} \zeta_5^{a_1^2(x_3^2+2x_4^2)+(1-a_1^4)x_3x_4^3-a_3x_3-a_4x_4}\enspace.
\end{align*}
Now consider two cases, namely $a_1=0$ and $a_1 \neq 0$.

\textbf{Case 1.} If $a_1=0$ then
\begin{align*}
\cW_f(0, a_2, a_3, a_4) &= 5 \sum_{x_4} \zeta_5^{-a_4x_4} \sum_{x_3} \zeta_5^{x_3(x_4^3-a_3)}= 25\zeta_5^{-a_3^3a_4}\enspace.
\end{align*}

\textbf{Case 2.} If $a_1\neq 0$ then 
\begin{align*}
\cW_f(a_1, a_2, a_3, a_4) &= 5\zeta_5^{ -a_1^3 a_2} \sum_{x_3} \sum_{x_4} \zeta_5^{a_1^2(x_3^2+2x_4^2)-a_3x_3-a_4x_4}\\
&= 5\zeta_5^{ -a_1^3 a_2} \sum_{x_3} \zeta_5^{a_1^2x_3^2-a_3x_3} \sum_{x_4} \zeta_5^{2a_1^2 x_4^2-a_4x_4}\\
&= -25 \zeta_5^{-a_4^2/3a_1^2+a_3^2/a_1^2 - a_1^3 a_2}\\
&= -25 \zeta_5^{a_1^2(a_3^2+3a_4^2)-a_1^3 a_2}\enspace,
\end{align*}
where \cite[Corollary~6]{HeKh06_1} is used to calculate the Walsh transform of square power functions. 

Thus, the dual of $f$ is 
\[
f^*(a_1,a_2,a_3,a_4)= -a_1^3a_2+a_1^2(a_3^2+3a_4^2)-(1-a_1^4)a_3^3a_4\enspace.
\]
It is easy to observe that the Walsh transform of $f^*$ at $(b_1, b_2, b_3, b_4) \in \F_5^4$ is given by
\begin{align*}
\cW_{f^*}(b_1, b_2, b_3, b_4) &= \sum_{a_1} \sum_{a_2} \sum_{a_3} \sum_{a_4} \zeta_5^{ -a_1^3a_2+a_1^2(a_3^2+3a_4^2)-(1-a_1^4)a_3^3a_4 - a_1b_1-a_2b_2-a_3b_3-a_4b_4}\\
&= \sum_{a_1} \sum_{a_3} \sum_{a_4} \zeta_5^{ a_1^2(a_3^2+3a_4^2)-(1-a_1^4)a_3^3a_4 - a_1b_1-a_3b_3-a_4b_4} \sum_{a_2} \zeta_5^{a_2(-a_1^3 - b_2)}\\
&= 5\sum_{a_1^3=-b_2} \sum_{a_3} \sum_{a_4}  \zeta_5^{ a_1^2(a_3^2+3a_4^2)-(1-a_1^4)a_3^3a_4 - a_1b_1-a_3b_3-a_4b_4} \\
&= 5\sum_{a_3} \sum_{a_4}\zeta_5^{ b_2^2(a_3^2+3a_4^2)-(1-b_2^4)a_3^3a_4 + b_1b_2^3-a_3b_3-a_4b_4} \\
&= 5\zeta_5^{ b_1b_2^3} \sum_{a_3} \sum_{a_4}\zeta_5^{ b_2^2(a_3^2+3a_4^2)-(1-b_2^4)a_3^3a_4 -a_3b_3-a_4b_4}\enspace.
\end{align*}
Now consider two cases, namely $a_1=0$ and $a_1 \neq 0$.

\textbf{Case 1.} If $b_2=0$ then
\begin{align*}
\cW_{f^*}(b_1, 0, b_3, b_4) = 5\zeta_5^{ b_1b_2^3} \sum_{a_3} \sum_{a_4}\zeta_5^{ -a_3^3a_4 -a_3b_3-a_4b_4}=5\zeta_5^{ b_1b_2^3} \sum_{a_3} \zeta_5^{ -a_3b_3} \sum_{a_4}\zeta_5^{ -(a_3^3 +b_4)a_4} =5\zeta_5^{ b_1b_2^3+b_4^3b_3} \enspace.
\end{align*}

\textbf{Case 2.} If $b_2\neq 0$ then 
\begin{align*}
\cW_{f^*}(b_1, b_2, b_3, b_4) &=5\zeta_5^{ b_1b_2^3} \sum_{a_3} \sum_{a_4}\zeta_5^{ b_2^2(a_3^2+3a_4^2) -a_3b_3-a_4b_4} \\
&= 5\zeta_5^{ b_1b_2^3} \sum_{a_3} \zeta_5^{b_2^2a_3^2-a_3b_3} \sum_{a_4} \zeta_5^{3b_2^2 a_4^2-a_4b_4}\\
&= -25 \zeta_5^{b_4^2/b_2^2+b_3^2/b_2^2 +b_1b_2^3}\\
&= -25 \zeta_5^{b_4^2b_2^2+b_3^2b_2^2 +b_1b_2^3} \enspace.
\end{align*}

Note that each monomial in $f$ and $f^*$ has algebraic degree $4$ or $8$. This implies that both $f$ and $f^*$ are $1$-$\cPS^{-}$ as for 
any $(x_1,x_2,x_3,x_4)\neq (0,0,0,0)$ and $\lambda\in\F_5^*$, we have $f(x_1,x_2,x_3,x_4)=f(\lambda x_1,\lambda x_2,\lambda x_3,\lambda 
x_4)$ and the same holds for $f^*$. Also, since both $f$ and $f^*$ are non weakly regular, they can not be in $2$-$\cPS^{-}$.\qed 
\end{proof}

\section{$p$-ary $m$-$\cPS^{+}$ class}
 \label{sec:PSplus}
Following is a generalization of \cite[Theorem 3.6]{LiLu14} that introduced $p$-ary $\cPS^{+}$ class, we suggest a method for constructing 
bent functions $f:\F_p^{2n}\rightarrow\F_p$ using partial $m$-spreads of $\F_p^{2n}$. 

\begin{theorem}
 \label{th:PSplus}
Let $p$ be an odd prime and $m,n$ be positive integers such that $m<2n$ and $p^m>3$. Assume $\cS$ is a partial $m$-spread of $\fp^{2n}$ and 
$f:\fp^{2n}\mapsto\fp$ is such that for each $T\in\cS$, $f$ is constant on $T^*$. Moreover, denote $D_j=f^{-1}(j)$ for $j\in\F_p$ and 
suppose that $D_0=\fp^{2n}\setminus\bigcup_{T\in\cS}T$ and $f(0)=z\neq 0$ (thus, $f$ is $m$-normal). Then $f$ is bent if and only if 
following four conditions hold 
\begin{enumerate}
\item $m\mid n$
\item $D_z$ is a union of exactly $\frac{(p^{n-1}+1)(p^n-1)}{p^m-1}$ elements of $\cS$
\item $D_j\cup\{0\}$ is a union of exactly $\frac{p^{n-1}(p^n-1)}{p^m-1}$ elements of $\cS$ for every $j\in\fp^*\setminus\{z\}$
\item for every nonzero $a\in\fp^{2n}$, $a^{\perp}$ contains either 
\begin{align}
\label{eq:T2bentcond1}
&\begin{cases}
N+\frac{p^{n-m}-1}{p^m-1}\ \mbox{subspaces}\ T\subseteq D_z\ \mbox{and}\\
N\ \mbox{subspaces}\ T\subseteq D_j \cup \{0\} \ \mbox{for remaining}\ j\in\fp^*\setminus\{z\}\\
\end{cases} 
\\\nonumber\mbox{or}\\
\label{eq:T2nwrbentcond1}
&\begin{cases}
M+\frac{p^{n-m}-1}{p^m-1}\ \mbox{subspaces}\ T\subseteq D_z\ \mbox{and}\\
M\ \mbox{subspaces}\ T\subseteq D_j\cup \{0\}\ \mbox{for remaining}\ j\in\fp^*\setminus\{z\}\\
\end{cases} 
\\\nonumber\mbox{or}\\
\label{eq:T2bentcond2}
&\begin{cases}
N+\frac{p^n-1}{p^m-1}\ \mbox{subspaces}\ T\subseteq D_z\ \mbox{and}\\
N\ \mbox{subspaces}\ T\subseteq D_j \cup \{0\} \ \mbox{for remaining}\ j\in\fp^*\setminus\{z\}\\
\end{cases} 
\\\nonumber\mbox{or}\\
\label{eq:T2nwrbentcond2}
&\begin{cases}
M-\frac{p^n-2p^{n-m}+1}{p^m-1}\ \mbox{subspaces}\ T\subseteq D_z\ \mbox{and}\\
M\ \mbox{subspaces}\ T\subseteq D_j\cup \{0\}\ \mbox{for remaining}\ j\in\fp^*\setminus\{z\}\\
\end{cases} 
\\\nonumber\mbox{or}\\
 \label{eq:T2bentcond3}
&\begin{cases}
N+\frac{p^{n-m}-1}{p^m-1}\ \mbox{subspaces}\ T\subseteq D_z\ \mbox{and}\\
N+p^{n-m}\ \mbox{subspaces}\ T\subseteq D_t\cup \{0\}\ \mbox{for some}\ t\in\fp^*\ \mbox{and}\\
N\ \mbox{subspaces}\ T\subseteq D_j\cup \{0\}\ \mbox{for remaining}\ j\in\fp^*\backslash\{z,t\}\enspace,\\
\end{cases}
\\\nonumber\mbox{or}\\
 \label{eq:T2nwrbentcond3}
&\begin{cases}
M+\frac{p^{n-m}-1}{p^m-1}\ \mbox{subspaces}\ T\subseteq D_z\ \mbox{and}\\
M-p^{n-m}\ \mbox{subspaces}\ T\subseteq D_t\cup \{0\}\ \mbox{for some}\ t\in\fp^*\ \mbox{and}\\
M\ \mbox{subspaces}\ T\subseteq D_j\cup \{0\} \ \mbox{for remaining}\ j\in\fp^*\backslash\{z,t\}\enspace,\\
\end{cases}
\end{align}
where $N=\frac{p^{n-1}(p^{n-m}-1)}{p^m-1}$ and $M=\frac{p^{n-m-1}(p^{n}+p^m-2)}{p^m-1}$ (here assume $m\neq n$). Also, 
$\cW_f(0)=p^n\zeta_p^z$, if \eqref{eq:T2bentcond1} holds then $\cW_f(a)=p^n$, if \eqref{eq:T2nwrbentcond1} holds then $\cW_f(a)=-p^n$, if 
\eqref{eq:T2bentcond2} holds then $\cW_f(a)=p^n\zeta_p^z$, if \eqref{eq:T2nwrbentcond2} holds then $\cW_f(a)=-p^n\zeta_p^z$, if 
\eqref{eq:T2bentcond3} holds then $\cW_f(a)=p^n\zeta_p^t$, and if \eqref{eq:T2nwrbentcond3} holds then $\cW_f(a)=-p^n\zeta_p^t$. 
Therefore, $f$ is bent then it is either regular or non weakly regular bent function and the latter holds only if $m\neq n$ and 
\eqref{eq:T2nwrbentcond1}, \eqref{eq:T2nwrbentcond2} or \eqref{eq:T2nwrbentcond3} is satisfied for at least one nonzero $a\in\F_p^{2n}$. 
\end{enumerate}
\end{theorem}
 
\begin{proof}
For $j\neq 0$ let $N_j\geq 0$ be such that $D_j=\bigcup_{i=1}^{N_j}\cS_{ji}^*$ if $j\neq z$, and $D_z=\bigcup_{i=1}^{N_z}\cS_{ji}$, where 
$\cS_{ji}$ are pairwise distinct (except $0$) elements of $\cS$ (if $N_j=0$ then $D_j=\emptyset$). Therefore, 
\begin{align*}
\lvert D_j\rvert&=N_j(p^m-1)\ \mbox{for}\ j\not\in\{0,z\}\\
\lvert D_z\rvert&=N_z(p^m-1)+1\\
\lvert D_0\rvert&=p^{2n}-(p^m-1)\sum_{j=1}^{p-1}N_j-1\enspace. 
\end{align*}
Note that $p-1$ divides $\lvert D_j\rvert$ when $j\neq z$ and does not divide $\lvert D_z\rvert$.

First, assume that $f$ is bent. Considering Lemma~\ref{le:Ny}, and the mentioned above divisibility, we conclude that $k=z$. By the same 
lemma, $\lvert D_j\rvert=p^{n-1}(p^n\mp 1)$ for $j\neq z$ and $\lvert D_z\rvert=p^{2n-1}\pm(p-1)p^{n-1}$. For $j\not\in\{0,z\}$, $\lvert 
D_j\rvert$ is divisible by $p^m-1$ and since $\gcd(p^n+1,p^m-1)\neq p^m-1$ (unless $p^m=3$), it follows that $\lvert D_0\rvert=\lvert 
D_j\rvert=p^{n-1}(p^n-1)$ and $m\mid n$. Thus, $N_j=\frac{p^{n-1}(p^n-1)}{p^m-1}$ for $j\neq z$. Also, $N_z=\frac{(p^{n-1}+1)(p^n-1)}{p^m-1}$ since $\lvert D_z\rvert=p^{2n-1}+(p-1)p^{n-1}$. Now calculate
\begin{align*}
\cW_f(0)&=\sum_{x\in\fp^{2n}}\zeta_p^{f(x)}=\sum_{j=0}^{p-1}\sum_{x\in D_j}\zeta_p^j\\
&=\lvert D_0\rvert+\lvert D_z\rvert\zeta_p^z+\sum_{j\in\fp^*\setminus\{z\}}\lvert D_j\rvert\zeta_p^j\\
&=p^{n-1}(p^n-1)+(p^{n-1}(p^n-1)+p^n)\zeta_p^z+ p^{n-1}(p^n-1)\sum_{j\in \fp^* \setminus\{z\}}\zeta_p^j\\
&=p^n\zeta_p^z+p^{n-1}(p^n-1)\sum_{j\in\fp}\zeta_p^j\\
&=p^n\zeta_p^z\enspace.
\end{align*}

From here, we conclude that $f$ is either regular or non weakly regular bent function. Now we show that \eqref{eq:T2bentcond1}, 
\eqref{eq:T2nwrbentcond1}, \eqref{eq:T2bentcond2}, \eqref{eq:T2nwrbentcond2}, \eqref{eq:T2bentcond3} or \eqref{eq:T2nwrbentcond3} holds. 
Let $a\in\fp^{2n}$ with $a\neq 0$. Assume that $a^{\perp}$ contains $s$ distinct subspaces from $\cS$, say $\cS_{j_1 i_1},\ldots ,\cS_{j_s 
i_s}$, where $0<j_1,\ldots,j_s<p$ and $1\leq i_1,\ldots,i_s\leq\frac{(p^{n-1}+1)(p^n-1)}{p^m-1}$. Consider 

\begin{align} \label{eq:T2wa}
\nonumber \cW_f(a)&=\sum_{x\in\fp^{2n}}\zeta_p^{f(x)-\langle a,x \rangle}\\
\nonumber &=\sum_{x\in D_0}\zeta_p^{-\langle a,x\rangle}+\zeta_p^z\sum_{x\in D_z}\zeta_p^{-\langle a,x\rangle}+\sum_{j\in\fp^*\setminus\{z\}}\zeta_p^j\sum_{x\in D_j}\zeta_p^{-\langle a,x \rangle}\\
\nonumber &=-\sum_{x\in D_z}\zeta_p^{-\langle a,x\rangle}-\sum_{j\in\fp^*\setminus\{z\}}\sum_{x\in D_j}\zeta_p^{-\langle a,x\rangle}+\zeta_p^z\sum_{x\in D_z}\zeta_p^{-\langle a,x\rangle}+\sum_{j\in\fp^*\setminus\{z\}}\zeta_p^j\sum_{x\in D_j}\zeta_p^{-\langle a,x\rangle}\\
\nonumber &=(p-1)\frac{p^{n-1}(p^n-1)}{p^m-1}+\frac{p^n-p^m}{p^m-1}-sp^m +\frac{p^{n-1}(p^n-1)}{p^m-1}+p^m(\zeta_p^{j_1}+\cdots+\zeta_p^{j_s})-\left(\frac{p^n-p^m}{p^m-1}\right)\zeta_p^z\\
\nonumber &=\frac{p^{n}(p^n-1)}{p^m-1}+p^m(\zeta_p^{j_1}+\cdots+\zeta_p^{j_s}-s)+\frac{p^n-p^m}{p^m-1}(1-\zeta_p^z)\\
&=\frac{p^{2n}-p^m}{p^m-1}+p^m(\zeta_p^{j_1}+\cdots+\zeta_p^{j_s}-s)-\frac{p^n-p^m}{p^m-1}\zeta_p^z\enspace.
\end{align}
Here, we used Lemma~\ref{le:sspsum} consistently. Since the dimension is even and $f$ is either regular or non weakly regular bent 
function, from \eqref{eq:bent_Wtc}, we have $\cW_f(a)= \pm p^n \zeta_p^t$ for some $t\in\fp$. We shall now consider following four cases. 

\textbf{Case 1.} Let $\cW_f(a)= p^n$. In this case \eqref{eq:T2wa} reduces to
\begin{equation}
 \label{eq:T2C1}
\zeta_p^{j_1}+\cdots+\zeta_p^{j_s}-\frac{p^{n-m}-1}{p^m-1}\zeta_p^z=s+p^{n-m}-\frac{p^{2n-m}-1}{p^m-1}\enspace.
\end{equation}
The LHS of \eqref{eq:T2C1} is an integer if and only if it is of the form $\ell(\zeta_p+\cdots+\zeta_p^{p-1})=-\ell$ for some positive 
integer $\ell$. Thus, 
\[\zeta_p^{j_1}+\cdots +\zeta_p^{j_s}=\ell(\zeta_p+\cdots+\zeta_p^{p-1})+ \frac{p^{n-m}-1}{p^m-1}\zeta_p^z\]
hence $s=\ell(p-1)+ \frac{p^{n-m}-1}{p^m-1}$. Inserting this value in \eqref{eq:T2C1} we obtain
\[-\ell=\ell(p-1)+\frac{p^{n-m}-1}{p^m-1}+p^{n-m}-\frac{p^{2n-m}-1}{p^m-1}\implies \ell=\frac{p^{n-1}(p^{n-m}-1)}{p^m-1}=N\enspace.\]
We conclude that $a^{\perp}$ contains $N$ distinct $\cS_{ji}\subseteq D_j \cup \{0\}$ for each $j\in\fp^*\setminus\{z\}$ and also contains 
$N+\frac{p^{n-m}-1}{p^m-1}$ distinct $\cS_{zi}\subseteq D_z$. Which is precisely condition \eqref{eq:T2bentcond1}.

\textbf{Case 2.} Let $\cW_f(a)= -p^n$. In this case \eqref{eq:T2wa} reduces to
\begin{equation}
 \label{eq:T2C2}
\zeta_p^{j_1}+\cdots+\zeta_p^{j_s}-\frac{p^{n-m}-1}{p^m-1}\zeta_p^z=s-p^{n-m}-\frac{p^{2n-m}-1}{p^m-1}\enspace.
\end{equation}
The LHS of \eqref{eq:T2C2} is an integer if and only if it is of the form $\ell(\zeta_p+\cdots+\zeta_p^{p-1})=-\ell$ for some positive 
integer $\ell$. Thus, 
\[\zeta_p^{j_1}+\cdots +\zeta_p^{j_s}=\ell(\zeta_p+\cdots+\zeta_p^{p-1})+ \frac{p^{n-m}-1}{p^m-1}\zeta_p^z\]
hence $s=\ell(p-1)+ \frac{p^{n-m}-1}{p^m-1}$. Inserting this value in \eqref{eq:T2C2} we obtain
\[-\ell=\ell(p-1)+\frac{p^{n-m}-1}{p^m-1}-p^{n-m}-\frac{p^{2n-m}-1}{p^m-1}\implies \ell=\frac{p^{n-m-1}(p^n+p^m-2)}{p^m-1}=M\enspace.\]
We conclude that $a^{\perp}$ contains $M$ distinct $\cS_{ji}\subseteq D_j \cup \{0\}$ for each $j\in\fp^*\setminus\{z\}$ and also contains 
$M+\frac{p^{n-m}-1}{p^m-1}$ distinct $\cS_{zi}\subseteq D_z$. Which is precisely condition \eqref{eq:T2nwrbentcond1}.

\textbf{Case 3.} Let $\cW_f(a)= p^n \zeta_p^z$. In this case \eqref{eq:T2wa} reduces to
\begin{equation}
 \label{eq:T2C3}
\zeta_p^{j_1}+\cdots+\zeta_p^{j_s}-\frac{p^n-1}{p^m-1}\zeta_p^z=s -\frac{p^{2n-m}-1}{p^m-1}\enspace.
\end{equation}
The LHS of \eqref{eq:T2C3} is an integer if and only if it is of the form $\ell(\zeta_p+\cdots+\zeta_p^{p-1})=-\ell$ for some positive 
integer $\ell$. Thus, 
\[
\zeta_p^{j_1}+\cdots +\zeta_p^{j_s}=\ell(\zeta_p+\cdots+\zeta_p^{p-1})+ \frac{p^n-1}{p^m-1}\zeta_p^z
\]
hence $s=\ell(p-1)+ \frac{p^n-1}{p^m-1}$. Inserting this value in \eqref{eq:T2C3} we obtain
\[
-\ell=\ell(p-1)+\frac{p^n-1}{p^m-1}-\frac{p^{2n-m}-1}{p^m-1}\implies \ell=\frac{p^{n-1}(p^{n-m}-1)}{p^m-1}=N\enspace.
\]
We conclude that $a^{\perp}$ contains $N$ distinct $\cS_{ji}\subseteq D_j \cup \{0\}$ for each $j\in\fp^*\setminus\{z\}$ and contains 
$N+\frac{p^n-1}{p^m-1}$ distinct $\cS_{zi}\subseteq D_z$. Which is precisely condition \eqref{eq:T2bentcond2}.

\textbf{Case 4.} Let $\cW_f(a)= -p^n \zeta_p^z$. In this case \eqref{eq:T2wa} reduces to
\begin{equation}
 \label{eq:T2C4}
\zeta_p^{j_1}+\cdots+\zeta_p^{j_s}+\frac{p^n-2p^{n-m}+1}{p^m-1}\zeta_p^z=s -\frac{p^{2n-m}-1}{p^m-1}\enspace.
\end{equation}
The LHS of \eqref{eq:T2C4} is an integer if and only if it is of the form $\ell(\zeta_p+\cdots+\zeta_p^{p-1})=-\ell$ for some positive 
integer $\ell$. Thus, 
\[
\zeta_p^{j_1}+\cdots +\zeta_p^{j_s}=\ell(\zeta_p+\cdots+\zeta_p^{p-1})-\frac{p^n-2p^{n-m}+1}{p^m-1}\zeta_p^z
\]
hence $s=\ell(p-1)-\frac{p^n-2p^{n-m}+1}{p^m-1}$. Inserting this value in \eqref{eq:T2C4} we obtain
\[
-\ell=\ell(p-1)-\frac{p^n-2p^{n-m}+1}{p^m-1}-\frac{p^{2n-m}-1}{p^m-1}\implies \ell=\frac{p^{n-m-1}(p^n+p^m-2)}{p^m-1}=M\enspace.
\]
We conclude that $a^{\perp}$ contains $M$ distinct $\cS_{ji}\subseteq D_j \cup \{0\}$ for each $j\in\fp^*\setminus\{z\}$ and contains 
$M-\frac{p^n-2p^{n-m}+1}{p^m-1}$ distinct $\cS_{zi}\subseteq D_z$. Which is precisely condition \eqref{eq:T2nwrbentcond2}.

\textbf{Case 5.} Let $\cW_f(a)= p^n \zeta_p^t$ for some $t \in \F_p \not\in \{0,z\}$. In this case \eqref{eq:T2wa} reduces to
\begin{equation}
 \label{eq:T2C5}
\zeta_p^{j_1}+\cdots+\zeta_p^{j_s}-\frac{p^{n-m}-1}{p^m-1}\zeta_p^z-p^{n-m} \zeta_p^t=s -\frac{p^{2n-m}-1}{p^m-1}\enspace.
\end{equation}
The LHS of \eqref{eq:T2C5} is an integer if and only if it is of the form $\ell(\zeta_p+\cdots+\zeta_p^{p-1})=-\ell$ for some positive 
integer $\ell$. Thus, 
\[
\zeta_p^{j_1}+\cdots +\zeta_p^{j_s}=\ell(\zeta_p+\cdots+\zeta_p^{p-1})+ \frac{p^{n-m}-1}{p^m-1}\zeta_p^z+p^{n-m} \zeta_p^t
\]
hence $s=\ell(p-1)+ \frac{p^{n-m}-1}{p^m-1}+p^{n-m}$. Inserting this value in \eqref{eq:T2C5} we obtain
\[
-\ell=\ell(p-1)+\frac{p^{n-m}-1}{p^m-1}+p^{n-m}-\frac{p^{2n-m}-1}{p^m-1}\implies \ell=\frac{p^{n-1}(p^{n-m}-1)}{p^m-1}=N\enspace.
\]
We conclude that $a^{\perp}$ contains $N$ distinct $\cS_{ji}\subseteq D_j \cup \{0\}$ for each $j\in\fp^*\setminus\{t,z\}$, contains 
$N+\frac{p^{n-m}-1}{p^m-1}$ distinct $\cS_{zi}\subseteq D_z$ and contains $N+p^{n-m}$ distinct $\cS_{ti}\subseteq D_t \cup \{0\}$. Which is 
precisely condition \eqref{eq:T2bentcond3}. 

\textbf{Case 6.} Let $\cW_f(a)= -p^n \zeta_p^t$ for some $t \in \F_p \not\in \{0,z\}$. In this case \eqref{eq:T2wa} reduces to
\begin{equation}
 \label{eq:T2C6}
\zeta_p^{j_1}+\cdots+\zeta_p^{j_s}-\frac{p^{n-m}-1}{p^m-1}\zeta_p^z+p^{n-m} \zeta_p^t=s -\frac{p^{2n-m}-1}{p^m-1}\enspace.
\end{equation}
The LHS of \eqref{eq:T2C6} is an integer if and only if it is of the form $\ell(\zeta_p+\cdots+\zeta_p^{p-1})=-\ell$ for some positive 
integer $\ell$. Thus, 
\[
\zeta_p^{j_1}+\cdots +\zeta_p^{j_s}=\ell(\zeta_p+\cdots+\zeta_p^{p-1})+ \frac{p^{n-m}-1}{p^m-1}\zeta_p^z-p^{n-m} \zeta_p^t
\]
hence $s=\ell(p-1)+ \frac{p^{n-m}-1}{p^m-1}-p^{n-m}$. Inserting this value in \eqref{eq:T2C6} we obtain
\[
-\ell=\ell(p-1)+\frac{p^{n-m}-1}{p^m-1}-p^{n-m}-\frac{p^{2n-m}-1}{p^m-1}\implies \ell=\frac{p^{n-m-1}(p^n+p^m-2)}{p^m-1}=M\enspace.
\]
We conclude that $a^{\perp}$ contains $M$ distinct $\cS_{ji}\subseteq D_j \cup \{0\}$ for each $j\in\fp^*\setminus\{t,z\}$, contains 
$M+\frac{p^{n-m}-1}{p^m-1}$ distinct $\cS_{zi}\subseteq D_z$ and contains $M-p^{n-m}$ distinct $\cS_{ti}\subseteq D_t \cup \{0\}$. Which is 
precisely condition \eqref{eq:T2nwrbentcond3}. 

In the opposite direction, the proof is straightforward.\qed 
\end{proof}

The following corollary is \cite[Theorem~3.6]{LiLu14} and for odd $p$, is a special case of Theorem~\ref{th:PSplus} with $m=n$. The case 
when $p=2$ is proven in \cite[Ch.~6]{Di74}. 

\begin{corollary}
Let $p$ be an odd prime and $n$ be positive integers such that $p^m>3$. Assume $\cS$ is a partial $n$-spread of $\fp^{2n}$ and 
$f:\fp^{2n}\rightarrow\fp$ is such that for each $T\in\cS$, $f$ is constant on $T^*$. Moreover, suppose that 
$f^{-1}(0)=\fp^{2n}\setminus\bigcup_{T\in\cS}T$ and $f(0)=z\neq 0$. Then $f$ is bent if and only if $f^{-1}(j)\cup\{0\}$ is a union of 
exactly $p^{n-1}$ elements of $\cS$ for every $j\in\fp^*\setminus\{z\}$ and $f^{-1}(z)$ is a union of exactly $p^{n-1}+1$ elements of 
$\cS$. Moreover, if $f$ is bent then for every nonzero $a\in\fp^{2n}$, $a^{\perp}$ contains either 
\begin{align}
 \label{eq:bentcond7}
&\mbox{none of the subspaces}\ T\in\cS\ \mbox{or}\\
 \label{eq:bentcond8}
&\mbox{one subspace}\ T \subseteq f^{-1}(t)\ \mbox{for some}\ t\in\fp^*
\end{align}
and these are the only cases possible. Also, $\cW_f(0)=p^n\zeta_p^z$ and if \eqref{eq:bentcond7} holds then $\cW_f(a)=p^n$ and if 
\eqref{eq:bentcond8} holds then $\cW_f(a)=p^n\zeta_p^t$. If $f$ is bent then it is a regular bent functions. 
\end{corollary}

\begin{definition}
The class of (regular) bent functions satisfying conditions of Theorem~\ref{th:PSplus} is called $p$-ary $m$-$\cPS^{+}$ class.
\end{definition}

\begin{example}
Define a quinary function $f:\F_5^4\mapsto\F_5$ as follows
\[f(x_1,x_2,x_3,x_4)=1+x_1x_2^3+x_3x_4^3\enspace.\]
Then $f$ is a regular bent function of algebraic degree four belonging to $1$-$\cPS^{+}$ class, so it is not from the $\cPS^{+}$ class.
\end{example}

\begin{example}
Define a quinary function $f:\F_5^4\mapsto\F_5$ as follows
\[f(x_1, x_2, x_3, x_4) =1+ x_1x_2^3 + x_2^2(x_3^2+2x_4^2)+(1-x_2^4)x_3x_4^3\enspace.\]
Then, $f$ is a non weakly regular bent function of algebraic degree eight belonging to $1$-$\cPS^{+}$ class. 
\end{example}

\section{Generalized Dillon Bent Functions}
 \label{sec:Di_monom}
Take function $f:\F_{p^{2n}}\mapsto\fp$ given by a sum of traces of monomials $\alpha_j x^{t(p^m-1)}$ where $m$ divides $n$, 
$0<t<\frac{p^n-1}{p^m-1}$ and $\alpha_t\in\F_{p^{2n}}$. Then $f$ is said to consist of {\em generalized Dillon exponents}. Let $\cS$ be the 
Desarguesian $m$-spread of $\F_{p^{2n}}$ (it consists of multiplicative cosets of $\fpm$). Then for each $T\in\cS$, $f$ is constant on 
$T^*$ and conditions of Theorem~\ref{th:PSminus} are satisfied. If $f$ is bent then it belongs to the $m$-$\cPS^{-}$ class and we call it 
{\em generalized Dillon bent function}. 

%The trace may be taken either from $\F_{p^{2n}}$ or from some subfields thereof, as certain values of $j$ will guarantee that 
%$x^{j(p^m-1)}$ belongs to a proper subfield of $\F_{p^{2n}}$ and then the choices for $\alpha_j$ and the trace function are made 
%accordingly.

Here we address the question whether a monomial function with a generalized Dillon exponent can be bent (but not in the $n$-$\cPS^{-}$ 
class). We start by recalling the result that is used further. 

Let $q$ be a power of a prime $p$ and $\chi$ be the canonical additive character of $\fq$ defined by $\chi(a)=e^{2\pi i\Tr(a)/p}$ for 
$a\in\fq$. For any $a\in\fq$, define an {\em $k$-dimensional Kloosterman sum} as 
\[K_k(a)=\sum_{x_1,\dots,x_k\in\fq^*}\chi\left(x_1+\cdots+x_k+ax_1^{-1}\cdots x_k^{-1}\right)\enspace.\]
Distribution of the values of the $k$-dimensional Kloosterman sum is a major problem of algebraic geometry.

\begin{theorem}[\cite{Mo07}]
 \label{th:mKl_sum}
For any $a\in\fqk^*$,
\[\sum_{x\in\fqk^*}\chi\left(\Tr_k(ax^{q-1})\right)=(-1)^{k-1}(q-1)K_{k-1}(\N(a))\enspace,\] 
where $\Tr_k()$ and $\N(a)=a^{(q^k-1)/(q-1)}$ are the trace and the norm functions from $\fqk$ onto $\fq$. If the sum is over leaders of 
the multiplicative cosets of $\fq^*$ then 
\[\sum_{x\in\fqk^*/\fq^*}\chi\left(\Tr_k(ax^{q-1})\right)=(-1)^{k-1}K_{k-1}(\N(a))\enspace.\] 
\end{theorem}

Now, using Theorem~\ref{th:mKl_sum}, we check when a generalized Dillon monomial $p$-ary functions is in the $m$-$\cPS^{-}$ class. This is 
a generalization of the known classes of Boolean bent functions of Dillon \cite{Di74} and ternary bent functions from 
\cite[Theorem~2]{HeKh06_1}. 

\begin{theorem}
 \label{th:monom}
Let $p$ be an odd prime and $n$, $m$, $t$ be positive integers such that $m$ divides $2n$ and $\gcd(t,(p^{2n}-1)/(p^m-1))=1$. For any 
$a\in\F_{p^{2n}}^*$, define $p$-ary function $f:\F_{p^{2n}}\mapsto\fp$ as 
\[f(x)=\Tr_{2n}\big(a x^{t(p^m-1)}\big)\enspace.\]
Then the Walsh transform coefficients of $f$ are equal to 
\begin{align*}
\cW_a(0)&=1+(-1)^{k-1}(p^m-1)K_{k-1}(\N(a))\quad\mbox{and}\\
\cW_a(-1)&=1+(-1)^k K_{k-1}(\N(a))+p^m\sum_{x\in U}\zeta_p^{f(x)}\enspace,
\end{align*}
where $k=2n/m$ and $U=\{x\in\F_{p^{2n}}^*/\F_{p^m}^*\ |\ \Tr^{2n}_m(x)=0\}$. Assume $p^m>3$ and $f$ is bent, then $m$ divides $n$, 
\begin{equation}
 \label{eq:Kl_cond}
(-1)^{k-1}K_{k-1}(\N(a))=(p^n-1)/(p^m-1)\enspace, 
\end{equation}
and $f$ is regular from the $m$-$\cPS^{-}$ class. Moreover, $\cW_a(0)=p^n$ and
\[\cW_a(-1)=1-(p^n-1)/(p^m-1)+p^m\sum_{x\in U}\zeta_p^{f(x)}\enspace.\]
\end{theorem}

\begin{proof}
Any $i\in\{0,\dots,p^{2n}-2\}$ can be uniquely written as $i=d j_1+j_2$ for some $j_1\in\{0,\dots,p^m-2\}$ and $j_2\in\{0,\dots,d-1\}$, 
where $d=(p^{2n}-1)/(p^m-1)$. Select and fix a primitive element $\xi$ of $\F_{p^{2n}}$. Any $x\in\F_{p^{2n}}^*$ can be uniquely written as 
$\xi^i$. Then 
\[\Tr_{2n}\big(a x^{t(p^m-1)}+x\big)=\Tr_{2n}\big(a\xi^{t(p^m-1)j_2}+\xi^{d j_1+j_2}\big)=\Tr_{2n}\left(a\beta^{j_2}+\gamma^{j_1}\xi^{j_2}\right)\] 
and $\Tr_{2n}\big(a x^{t(p^m-1)}\big)=\Tr_{2n}\left(a\beta^{j_2}\right)$, where $\beta=\xi^{t(p^m-1)}$ and $\gamma=\xi^d\in\fpm^*$. 

The Walsh transform coefficient of $f$ evaluated at $-1$ is equal to 
\[\cW_a(-1)=\sum_{x\in\F_{p^{2n}}}\zeta_p^{\Tr_{2n}\left(a x^{t(p^m-1)}+x\right)}=1+\sum_{j_2=0}^{d-1}\zeta_p^{\Tr_{2n}\left(a\beta^{j_2}\right)}\sum_{j_1=0}^{p^m-2}\zeta_p^{\Tr_{2n}\left(\gamma^{j_1}\xi^{j_2}\right)}\enspace.\] 
Consider the latter sum 
\[\sum_{j_1=0}^{p^m-2}\zeta_p^{\Tr_{2n}\left(\gamma^{j_1}\xi^{j_2}\right)}=
\sum_{j_1=0}^{p^m-2}\zeta_p^{\Tr_m\left(\gamma^{j_1}\Tr^{2n}_m(\xi^{j_2})\right)}=
\left\{\begin{array}{ll}
p^m-1,&\ \mbox{if}\quad\Tr^{2n}_m(\xi^{j_2})=0\\
-1,&\ \mbox{otherwise}
\end{array}\right.\] 
since the multiplicative order of $\gamma$ is equal to $p^m-1$ and thus, $\gamma$ is a primitive element of $\fpm$ and $\gamma^{j_1}$ runs 
through all nonzero elements of $\fpm$. Further, $\Tr^{2n}_m(x)=0$ for $p^{2n-m}-1$ elements of $\F_{p^{2n}}^*$ that reduce to 
$(p^{2n-m}-1)/(p^m-1)$ elements $x\in\F_{p^{2n}}^*/\F_{p^m}^*$ (this number of values of $j_2$ satisfy $\Tr^{2n}_m(\xi^{j_2})=0$). $U$ 
denotes the set containing all such $x$. 

Therefore,
\begin{align*}
\cW_a(-1)&=1-\sum_{j_2=0}^{d-1}\zeta_p^{\Tr_{2n}\left(a\beta^{j_2}\right)}+p^m\sum_{\substack{j_2=0\\ \Tr^{2n}_m(\xi^{j_2})=0}}^{d-1}\zeta_p^{\Tr_{2n}\big(a\beta^{j_2}\big)}\\
&=1-\sum_{j=0}^{d-1}\zeta_p^{\Tr_{2n}\left(a\beta^{j}\right)}+p^m\sum_{x\in U}\zeta_p^{f(x)}\enspace.
\end{align*}
Note that the sets $\{\beta^{j}\ |\ j=0,\dots,d-1\}$ and $\{\xi^{j(p^m-1)}\ |\ j=0,\dots,d-1\}$ are equal because $t$ and 
$(p^{2n}-1)/(p^m-1)$ are coprime and thus, using Theorem~\ref{th:mKl_sum} with $q=p^m$ and $k=2n/m$, 
\begin{align*}
\cW_a(-1)&=1-\sum_{j=0}^{d-1}\zeta_p^{\Tr_{2n}\left(a\xi^{j(p^m-1)}\right)}+p^m\sum_{x\in U}\zeta_p^{f(x)}\\
&=1-\sum_{j=0}^{d-1}\chi\left(\Tr^{2n}_m\big(a\xi^{j(p^m-1)}\big)\right)+p^m\sum_{x\in U}\zeta_p^{f(x)}\\
&=1+(-1)^k K_{k-1}(\N(a))+p^m\sum_{x\in U}\zeta_p^{f(x)}\enspace.
\end{align*}

Further, for any $b\in\fpn^*$ we have
\begin{align*}
\cW_a(-b)&=S_{ab^{-t(p^m-1)}}(-1)\\
&=1+(-1)^k K_{k-1}(\N(a))+p^m\sum_{{\substack{j=0\\\Tr^{2n}_m(\xi^j)=0}}}^{d-1}\zeta_p^{\Tr_{2n}\big(ab^{-t(p^m-1)}\beta^{j}\big)}=\cW_a(b) 
\end{align*}
since $\N(ab^{-t(p^m-1)})=\big(ab^{-t(p^m-1)}\big)^{(p^{2n}-1)/(p^m-1)}=a^{(p^{2n}-1)/(p^m-1)}=\N(a)$. 

The Walsh transform coefficient of $f$ evaluated at $0$ is equal to 
\begin{align*}
\cW_a(0)&=\sum_{x\in\F_{p^{2n}}}\zeta_p^{\Tr_{2n}\left(ax^{t(p^m-1)}\right)}=1+(p^m-1)\sum_{j=0}^{d-1}\zeta_p^{\Tr_{2n}\left(a\beta^{j}\right)}\\
&=1+(-1)^{k-1}(p^m-1)K_{k-1}(\N(a))\enspace.
\end{align*}

If $p^n>3$ and $f$ is bent, then by Theorem~\ref{th:PSminus}, $m$ divides $n$ and $f$ is a regular bent function from $m$-$\cPS^{-}$ class, 
$\cW_a(0)=p^n$ that gives 
\[(-1)^{k-1}K_{k-1}(\N(a))=(p^n-1)/(p^m-1)\enspace.\] 
Then also
\[\cW_a(-1)=1-(p^n-1)/(p^m-1)+p^m\sum_{x\in U}\zeta_p^{f(x)}\enspace.\]
This completes the proof.\qed
\end{proof}

%\bibliographystyle{splncs04}
%\bibliography{IEEEabrv,mrabbrev,all}

\begin{thebibliography}{10}
\providecommand{\url}[1]{\texttt{#1}} \providecommand{\urlprefix}{URL } \providecommand{\doi}[1]{https://doi.org/#1} 

\bibitem{AnMei22} Anbar, N., Meidl, W.: Bent partitions. Des. Codes Cryptogr.  \textbf{90}(4), 
  1081--1101 (Apr 2022)

\bibitem{Ca20} Carlet, C.: Boolean Functions for Cryptography and Coding Theory. Cambridge 
  University Press, Cambridge (2020)

\bibitem{CeMeiPo13_1} \c{C}e\c{s}melio\u{g}lu, A., Meidl, W., Pott, A.: On the dual of (non)-weakly 
  regular bent functions and self-dual bent functions. Advances in Mathematics
  of Communications  \textbf{7}(4),  425--440 (Nov 2013)

\bibitem{CeMeiPo16} \c{C}e\c{s}melio\u{g}lu, A., Meidl, W., Pott, A.: There are infinitely many 
  bent functions for which the dual is not bent. {IEEE} Trans. Inf. Theory
  \textbf{62}(9),  5204--5208 (Sep 2016)

\bibitem{Di74} Dillon, J.F.: Elementary {H}adamard Difference Sets. Ph.D. thesis, University 
  of Maryland (1974)

\bibitem{HeKh06_1} Helleseth, T., Kholosha, A.: Monomial and quadratic bent functions over the 
  finite fields of odd characteristic. {IEEE} Trans. Inf. Theory
  \textbf{52}(5),  2018--2032 (May 2006)

\bibitem{HeKh10_4} Helleseth, T., Kholosha, A.: New binomial bent functions over the finite fields 
  of odd characteristic. {IEEE} Trans. Inf. Theory  \textbf{56}(9),  4646--4652
  (Sep 2010)

\bibitem{HeKhSp26} Helleseth, T., Kholosha, A., Spithaki, N.: Ternary binomial and trinomial bent 
  functions in the completed {M}aiorana-{M}c{F}arland class. arXiv:2507.20715v2
  (2026)

\bibitem{Ho04_1} Hou, X.D.: $p$-{A}ry and $q$-ary versions of certain results about bent 
  functions and resilient functions. Finite Fields Appl.  \textbf{10}(4),
  566--582 (Oct 2004)

\bibitem{KuScWe85} Kumar, P.V., Scholtz, R.A., Welch, L.R.: Generalized bent functions and their 
  properties. J. Combin. Theory Ser. A  \textbf{40}(1),  90--107 (Sep 1985)

\bibitem{LiLu14} Lison\v{e}k, P., Lu, H.Y.: Bent functions on partial spreads. Des. Codes 
  Cryptogr.  \textbf{73}(1),  209--216 (Oct 2014)

\bibitem{Mo07} Moisio, M.: On the number of rational points on some families of {F}ermat 
  curves over finite fields. Finite Fields Appl.  \textbf{13}(3),  546--562
  (Jul 2007)

\bibitem{Ny91} Nyberg, K.: Constructions of bent functions and difference sets. In: 
  Damg{\aa}rd, I.B. (ed.) Advances in Cryptology - EuroCrypt '90. Lecture Notes
  in Computer Science, vol.~473, pp. 151--160. Springer-Verlag, Berlin (1991)

\bibitem{OzPe20} \"{O}zbudak, F., Pelen, R.M.: Duals of non-weakly regular bent functions are 
  not weakly regular and generalization to plateaued functions. Finite Fields
  Appl.  \textbf{64}, 101668 (Jun 2020)

\end{thebibliography}

\end{document}